\documentclass[10pt, conference, letterpaper]{IEEEtran}

\usepackage{amsmath}
\usepackage[ruled,vlined,noend]{algorithm2e} 
\usepackage{subfigure} 
\usepackage{enumitem} 
\usepackage{amsfonts} 
\usepackage{amsthm}  
\usepackage{footnote} 
\usepackage{pgfplots} 
\pgfplotsset{compat=1.6}

\newtheorem{theorem}{Theorem}
\newtheorem{definition}{Definition}
\newtheorem{lemma}{Lemma}
\newtheorem{example}{Example}
\newtheorem{corollary}[theorem]{Corollary}

\newtheorem{remark}{Remark}

\newcommand{\ceil}[1]{\lceil #1 \rceil}
\newcommand{\floor}[1]{\lfloor #1 \rfloor}

\newcommand{\TRANSACT}[3] {\ensuremath{({#1} \to_{#2} {#3})}} 

\newcommand{\etal}{{\em et al.}}

\definecolor{myblue}{RGB}{80,80,160}
\definecolor{mygreen}{RGB}{80,160,80}

\begin{document}
\bstctlcite{IEEEexample:BSTcontrol} 

\title{On Ranges and Partitions in Optimal {TCAMs}}

\author{
\IEEEauthorblockN{Yaniv Sadeh}
\IEEEauthorblockA{Tel-Aviv University, Israel \\ yanivsadeh@mail.tau.ac.il}
}

\maketitle

\begin{abstract}
Traffic splitting is a required functionality in networks, for example for load balancing over paths or servers, or by the source's access restrictions. The capacities of the servers (or the number of users with particular access restrictions) determine the sizes of the parts into which traffic should be split. A recent approach implements traffic splitting within the ternary content addressable memory (TCAM), which is often available in switches. It is important to reduce the amount of memory allocated for this task since TCAMs are power consuming and are often also required for other tasks such as classification and routing. In the longest-prefix model (LPM), \cite{DravesKSZ99} finds a minimal representation of a function, and \cite{BitMatcher} finds a minimal representation of a partition.
In certain situations, range-functions are of special interest, that is, all the addresses with the same target, or action, are consecutive. In this paper we show that minimizing the amount of TCAM entries to represent a partition comes at the cost of fragmentation, such that for some partitions some actions must have multiple ranges. Then, we also study the case where each target must have a single segment of addresses.
\end{abstract}

\section{Introduction}
\label{section_indtroduction}
In many networking applications, traffic has to be split into multiple possible targets. For example, this is required in order to partition traffic among multiple paths to a destination based on link capacities, and when sending traffic to one of multiple servers proportionally to their CPU or memory resources for load balancing. Traffic splitting also arises in maintaining access-control lists (ACLs). Here, we want to limit the number of users with specific permissions. We do this by associating a fixed quota of $W$-bit identifiers with each ACL, and granting a particular access only to users that have one of these identifiers.
In general, we address any scenario where traffic should be split by allocating a particular subset of identifiers of a specified size to each part (a part could be associated with a server or an ACL or with some other object).

It is increasingly common to rely on network switches  to perform the split~\cite{al2008scalable},\cite{Ananta},\cite{DASH_alg}. Equal cost multipath routing (ECMP)~\cite{RFC2992} uses hashing on flows to uniformly select one of target values written as memory entries. 
WCMP~\cite{WCMP},\cite{Zegura} (Weighted ECMP) generalizes the selection for non-uniform selections through entry repetitions, implying a distribution according to the number of appearances of each possible target. The implementation of some distributions in WCMP may require a large hash table. While for instance implementing a 1:2 ratio can be done with three entries (one for the first target and two for the second), the implementation of a  ratio  of the form $1:2^W-1$ is expensive, requiring $2^W$ entries. Memory can grow quickly for particular distributions over many targets, even if they are only being approximated. A recent approach~\cite{DASH_alg} refrains from memory blowup by comparing the hash to range-boundaries. Since the hash is tested sequentially against each range, it restricts the total number of load-balancing targets.

Recently, a natural approach was taken to implement traffic splitting within the Ternary Content Addressable Memory (TCAM), available in commodity switch architectures. For some partitions this allows a much cheaper representation~\cite{WangBR11},\cite{Niagara},\cite{AccurateExp},\cite{BitMatcher}. In particular, a partition of the form $1:2^W-1$ can be implemented with only two entries. Unfortunately, TCAMs are power consuming and thus are of limited size~\cite{appelman2012performance},\cite{McKeownABPPRST08}. Therefore a common goal is to minimize the representation of a partition in TCAMs. Finding a representation of a partition becomes more difficult when the number of possible targets is large. Focusing on the \emph{Longest Prefix Match model} (LPM), \cite{Niagara} suggested an algorithm named \emph{Niagara}, and showed that it produces small representations in practice.
They also evaluated a tradeoff of reduced accuracy for less rules. \cite{BitMatcher} suggested an optimal algorithm named \emph{Bit Matcher} that computes a smallest TCAM for a target partition. They also proved that Niagara always computes a smallest TCAM, explaining its good empirical performance. The resulting size of the TCAM table is analyzed in \cite{SadehSOSRTCAMsize} and \cite{SadehISITTCAMsize}: \cite{SadehSOSRTCAMsize} provide a few lower and upper bounds, and \cite{SadehISITTCAMsize} focuses on average-case analysis when partitions are sampled uniformly from a certain distribution. A representation of a partition can be memory intensive when the number of possible targets is large. \cite{TCAM_Linf_TON} and \cite{TCAM_L1_INFOCOM} consider ways of finding approximate-partitions whose representation is cheaper than a desired input partition.

In this work we study the tradeoff between representing a partition with a smallest number of TCAM rules, versus allocating this partition by minimizing the number of ranges for each of the targets. Rottenstreich~\etal~\cite{AccurateExp} proved that any partition into $k=2$ parts can be implemented optimally on a TCAM such that a segment of addresses is assigned to each target. That is, for $P=[p,2^W-p]$ there is an optimal allocation such that the first $p$ addresses belong to the first target, and the last $2^W-p$ belong to the second.

{\bf Our Contributions.}
We show that in general the two objectives, of minimizing the number of rules and allocating in ranges, collide. More formally:

(1) We show that for $k \ge 3$ such an optimal allocation may not exist. Section~\ref{section_segments_simple} introduces an algorithm for constructing a smallest TCAM that assigns a single segment of addresses to each target, according to their order in the input. That is, the first $p_1$ addresses are allocated to the first target, the following $p_2$ addresses are allocated to the second target and so on. Then we use this algorithm and the Bit Matcher algorithm to show that the partition $P = [13,13,6]$ does not have a minimal representation such that a single segment is associated with each target. This partition is minimal in $W$ and $k$ among all partitions with this property. This implies that for any $k \ge 3$ such partitions exist. Indeed, we show how to ``embed'' this problematic instance inside partitions with more parts such that not all parts can be ranges.

(2) In Section~\ref{section_segments_generalized} we generalize the analysis to show that even if we allow the addresses of each target to be fragmented into at most $M < \frac{c+1}{4} + \frac{1}{2c}$ segments, then for any $k \ge c$ there are partitions such that addresses allocated to some parts must be fragmented into more than $M$ segments if realized by a minimal set of TCAM rules.

We note that while we prove the existence (constructively) of ``very bad'' partitions, our construction requires $W$ to be very large. Because of this, and based on additional exhaustive computer-search on small values of $k$ and $W$, we believe that in practice most partitions can be realized with minimal number of rules while mapping a small number of segments to each target.

(3) In Section~\ref{section_looking_for_best_order_of_segments} we focus on the case where a partition must be realized as a single segment of addresses per target, and show that an arbitrary order requires no more than $k$ times the number of rules compared to the best ordering. We also prove that by slightly optimizing our choice, we can reduce this factor to at most $\min(\frac{k+1}{3},W-\floor{\lg k} + 1)$.

The structure of the rest of the paper is as follows. In Section~\ref{section_model} we formally define the problem and set some terminology and definitions. In Section~\ref{section_segments_simple} we show an example to a partition  with $k=3$ parts that cannot be realized both with the minimum number of rules and as three segments of consecutive addresses. In Section~\ref{section_segments_generalized} we dive deeper and show the more general result on the tradeoff between minimizing the number of rules and segmentation. In Section~\ref{section_looking_for_best_order_of_segments} we revisit the question of finding the best order of segments that minimize the required number of rules, given this restriction. Section~\ref{section_conclusions} summarizes our work.

\section{Traffic Splitting Problem}
\label{section_model}

A Ternary Content Addressable Memory (TCAM) of width $W$ is a table of entries, or \emph{rules}, each containing a \emph{pattern} and a \emph{target}. We assume that each target is an integer in $\{1,\ldots,k\}$, and also define a special target $0$ for dealing with addresses that are not matched by any rule. Each pattern is of length $W$ and consists of bits (0 or 1) and don't-cares ($*$). An \emph{address} is said to match a pattern if all of the specified bits of the pattern (ignoring don't-cares) agree with the corresponding bits of the address. If several rules fit an address, the first rule applies. An address $v$ is associated with the target of the rule that applies to $v$.

The analysis in this paper follows the \emph{Longest Prefix Match (LPM)} model, restricting rule-patterns in the TCAM to include wildcards only as a suffix such that a pattern can be described by \emph{a prefix} of bits. 
This model is motivated by specialized hardware as in \cite{LPM_TCAM}, and is assumed in many previous studies \cite{WangBR11},\cite{Niagara},\cite{AccurateExp},\cite{BitMatcher}. Common programmable switch architectures such as RMT and Intel’s FlexPipe have tables  dedicated to LPM~\cite{NSDIJose15},\cite{Forwarding13},\cite{FlexPipe}. In general, much less is known about general TCAM rules.

There can be multiple ways to represent the same partition in a TCAM, as a partition does not restrict the particular addresses mapped to each target but only their number. For instance, with $W=3$ the rules $\{ \textsc{011} \to 1, \textsc{01*} \to 2, \textsc{0**} \to 3, \textsc{***} \to 1\}$ imply a partition [5,1,2]  of addresses mapped to each of the targets \{1,2,3\}. Similarly, the same partition can also be derived using only three rules $\{\textsc{000} \to 2, \textsc{01*} \to 3, \textsc{***} \to 1\}$ (although this changes the identity of the addresses mapped to each target).

Given a desired partition $P = [p_1,\ldots,p_k]$ of the whole address space of $2^W$ addresses of $W$-bits such that $p_i > 0$ addresses should reach target $i$ ($\sum_i p_i = 2^W$), we aim to compute a smallest set of TCAM rules that partition traffic according to $P$. Note that $k \le 2^W$, and all addresses are considered equal in this model.\footnote{The model assumes implicitly that every address is equally likely to arrive, therefore the TCAM implementation only requires each target to receive a certain number of addresses. This assumption might not hold in practice, but it can be mitigated by ignoring bits which are mostly fixed like subnet masks etc. For example, \cite{Kang2014NiagaraSL} analyzed traces of real-data and concluded that for those traces about $6{-}8$ bits out of the client's IPv4 address are ``practically uniform''.}

A TCAM $T$ can be identified with a sequence $s$ of transactions between targets, defined as follows. Start with an empty sequence, and consider the change in the mapping defined by $T$ when we delete the first rule of $T$, with target $i \in \{1,\ldots,k\}$. Following this deletion some of the addresses may change their mapping to a different target, or become unallocated. If by deleting this rule, $m$ addresses are re-mapped from $i$ to $j \in \{0,\ldots,k\}$ (recall that $j=0$ means unallocated), we add to $s$ a transaction $\TRANSACT{i}{m}{j}$. We then delete the next rule of $T$ and add the corresponding transactions to $s$, and continue until $T$ is empty and all addresses are unallocated.

\begin{definition}[Transactions] 
\label{def_transaction}
Denote a transaction of size $m$ from $p_i$ to $p_j$ by $\TRANSACT{i}{m}{j}$. Applying this transaction to a partition $P = [p_1,\ldots,p_k]$ updates its values as: $p_i \leftarrow p_i - m$, $p_j \leftarrow p_j + m$.
We refer to $i$ as the \emph{sender} and to $j$ as the \emph{receiver} (of the transaction).
\end{definition}

In the LPM model, a deletion of a single TCAM rule corresponds to exactly one transaction (or none if the rule was redundant), of size that is a power of $2$. 

\begin{example}
\label{example_tcam_to_sequence}
Consider the rules: $\{011 \to 1, 01{*} \to 2, 0{*}{*} \to 3, {*}{*}{*} \to 1 \}$ with $W=3$. They partition $2^W = 8$ addresses to $k=3$ targets. Deleting the first rule corresponds to the transaction $\TRANSACT{1}{1}{2}$. The deletion of each of the following three rules also corresponds  to a single transaction, $\TRANSACT{2}{2}{3}$, $\TRANSACT{3}{4}{1}$ and $\TRANSACT{1}{8}{0}$, respectively, see Fig.~\ref{figure_mapping_example}.

\begin{figure}[t]
  \centering
    \includegraphics[width=0.99\linewidth]{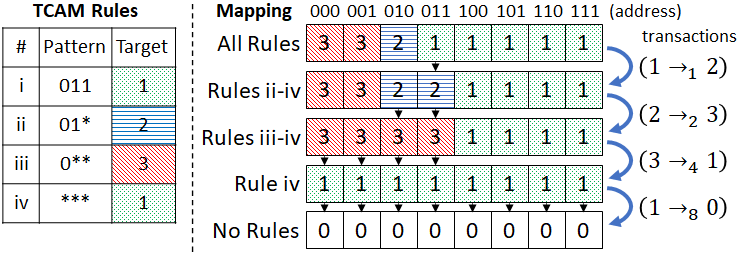}
  \caption{Example~\ref{example_tcam_to_sequence}: Rules and corresponding transactions due to remapping.}
  \label{figure_mapping_example}
\end{figure}
\end{example}

\begin{definition}[Complexity] 
\label{definition_length_of_partition}
Let $P=[p_1,\ldots,p_k]$ be a partition of $2^W$. We define $\lambda(P)$ to be the length of a shortest sequence of transactions of sizes that are powers of $2$, that zeroes $P$. This is also equal to the size of a smallest LPM TCAM that realizes $P$.
We say that $\lambda(P)$ is the \emph{complexity} of $P$. 
\end{definition}

\looseness=-1
It was shown in~\cite{BitMatcher} that Bit Matcher (see Algorithm~\ref{alg_bit_matcher})
computes a (shortest) sequence of transactions for an input partition $P$ whose length is $\lambda(P)$ which can also be mapped to a TCAM table. No smaller TCAM exists, as it will correspond to a shorter sequence in contradiction to the minimality of $\lambda(P)$.

\begin{algorithm}[t]
    \SetAlgoLined
    \DontPrintSemicolon

    \SetKwFunction{funcGenerateSequence}{Compute}
    \SetKwFunction{funcProcessLowLevel}{ProcessLevel}
    \SetKwProg{Fn}{Function}{:}{}
    
    { 
    \Fn{\funcGenerateSequence{partition $P$, $\sum_{i=1}^k p_i =2^W$}}{
        Initialize $s$ to be an empty sequence.\;
        \For{level $d = 0 \ldots W-1$} {
            $s' = \funcProcessLowLevel(P,d)$. Update $s \leftarrow s \cup s'$ and apply $s'$ on $P$.
        }
        \Return $s \cup \{ \TRANSACT{i}{2^W}{0} \ |\ 1 \le i \le k \wedge p_i = 2^W \}$.
        }
    \;
    \Fn{\funcProcessLowLevel{partition $P$, $d$}}{
        {\footnotesize // $p_i$ is \emph{bit lexicographic} smaller than $p_j$ if ${p_i}^r < {p_j}^r$ where $n^r$ is the bit-reverse of $n$ with respect to word size $W$.}
    
       Let $A = \{ i \mid i \ge 1 \wedge p_i[d] = 1\}$.\;
       
       Let $A_h \subset A$ consists of the $|A|/2$ indices of the $p_i$s that are largest in bit lexicographic order, and let $A_l = A \setminus A_h$. Pair the elements of $A_h$ and $A_l$ arbitrarily. For each pair $i \in A_l,j \in A_h$ append to $s'$ (initially $s' = \emptyset$) the transaction $\TRANSACT{i}{2^d}{j}$.\;
       Finally, \Return $s'$.
    }
    }
    \caption{Bit Matcher Algorithm~\cite{BitMatcher}}
    \label{alg_bit_matcher}
\end{algorithm}

To conclude this section, we introduce the ``visual terminology'' which will be used in the following section. A set of prefix rules $T$ (of an LPM TCAM) corresponds to a subset of the nodes of the full binary trie (see Fig.~\ref{figure_coloring_of_trie}). In particular, the match-all prefix corresponds to the root, and any other nonempty prefix $p$ corresponds to the node whose path from the root gives $p$ if we change an edge to a left child to $0$ and an edge to a right child to $1$. The rule which applies to an address $v$ (which is a leaf) corresponds to the closest ancestor of $v$ which represents a prefix in $T$. This is the longest prefix of the address $v$ in $T$. Henceforth in this paper we shall adopt this way to regard TCAMs. We will also identify each target with a color, and say that the nodes of the tree are colored accordingly: We color a node $v$ by the color of the target of the rule of the closest marked ancestor of $v$, and we define a conflict of colors as follows.

\begin{definition}[Trie Coloring Conflict]
\label{definition_color_conflict}
Let $C$ be some coloring of a trie. We say that a node $v$ is in \textbf{conflict} if the color of $v$ is different than the color of its parent, or if $v$ is the root. We also associate a non-root conflict with the edge between $v$ and its parent.
\end{definition}

A TCAM of $k$ targets defines a coloring of the trie with $k$ colors such that the number of rules equals the number of conflicts and vice versa. The size of each part in the partition induced by this TCAM is equal to the number of leaves of the corresponding color, i.e.\ $p_i$ leaves are colored by the $i$th color. Fig.~\ref{figure_coloring_of_trie} presents an example of a coloring of a trie and its associated conflicts. The smallest TCAM representing a partition $P$ corresponds to a coloring of the trie with the smallest number of conflicts.

\begin{figure}[t]
  \centering
  \includegraphics[width=0.9\linewidth]{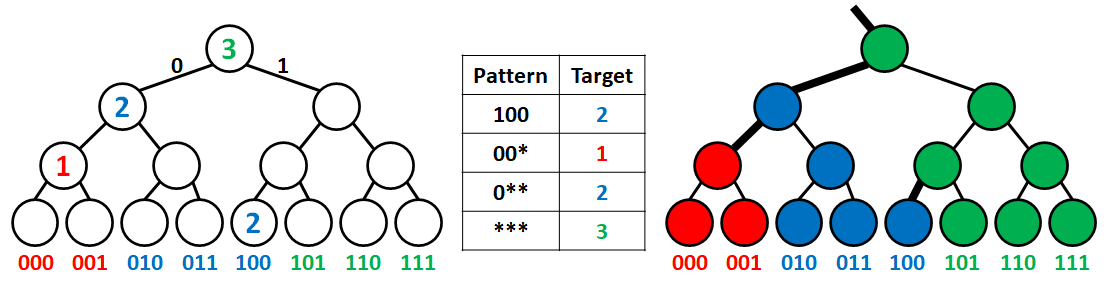}
  \caption{Coloring a trie is equivalent to marking its nodes which are associated with the TCAM rules. In the middle there is a TCAM table. On the left, a representation by a trie with marked nodes, and on the right a representation by a colored trie. Fat edges emphasize conflicts of colors, and the root also has a fictitious parent-edge in the figure to highlight its own conflict.}
  \label{figure_coloring_of_trie}
\end{figure}

\begin{remark}
\label{remark_transactions_in_coloring}
Performing a transaction is interpreted as unmarking a node in the tree. Furthermore, if we perform the transactions in bottom-up order (e.g. by increasing sizes) then by definition we always unmark a node without marked descendants. In coloring terminology, a transaction colors the subtree of a node in conflict by the color of its parent.
This subtree is monochromatic when we recolor it, if we perform transactions in bottom-up order.
\end{remark}

\section{Segments Allocation on a TCAM and an Impossibility-Result}
\label{section_segments_simple}

\subsection{Computing Rules for Segmented Allocation}

Given a partition of the leaves of the trie into $k$ segments of sizes $p_1,\ldots,p_k$, Algorithm~\ref{alg_consecutive_ranges} colors the trie with minimum conflicts such that each of these segments has a different color. We refer to such a coloring as an {\em optimal coloring.} We note in advance that this algorithm can be viewed as a special-case of the more general and optimal algorithm of Draves~\etal~\cite{DravesKSZ99}, called ORTC, which finds a coloring with minimum number of conflicts (corresponds to a TCAM of minimal size) for any given coloring of the leaves. We elaborate on this later and explain how exactly ORTC generalizes our algorithm.

The core idea of Algorithm~\ref{alg_consecutive_ranges} is that as long as $2$ sibling nodes have the same color, we can color their parent by the same color without conflicts. Only when the siblings are colored differently we must decide which of them will conflict with their parent. To do it properly, we only require local information of four nodes sharing a grandparent.

\begin{algorithm}[t]
    \SetAlgoLined
    \DontPrintSemicolon
    \KwIn{
        A coloring $\chi$ of the leaves of a trie of depth $W$, in $k$ consecutive segments of lengths $x_1,\ldots,x_k$ (such that $\sum_{i=1}^{k}{x_i = 2^W}$).
    }
    \KwOut{
        Trie coloring $\chi^*$ that extends $\chi$, with a minimum number of conflicts.
    }
    \hspace{2mm}

    Initialization: Color the leaves according to $\chi$.
    
    \For{trie depth = $W, \ldots, 2$} {
        For each 4-cousin nodes (same grandparent), whose colors in order are $a,b,c,d$:\\
        \ \ If $b=c$, color the parent of $a$ and $b$ by $b$, otherwise color it by $a$. Similarly, \\
        \ \ if $b=c$, color the parent of $c$ and $d$ by $b$, otherwise color it by $d$.
    }
    
    Color the root by the color of its right child.
    
    \caption{Minimum-conflict trie coloring for single-segment leaves coloring.}
    \label{alg_consecutive_ranges}
\end{algorithm}

\begin{theorem}[Time and space complexity of Algorithm~\ref{alg_consecutive_ranges}]
\label{theorem_runtime_consecutive_ranges}
Let $N$ be the number of conflicts in the coloring produced by Algorithm~\ref{alg_consecutive_ranges}. Then Algorithm~\ref{alg_consecutive_ranges} requires $O(k + N)$ space, and takes $O(Wk)$ time assuming $W$-bit word operations take $O(1)$ time.
\end{theorem}
\begin{proof}
It is possible to implement Algorithm~\ref{alg_consecutive_ranges} such that it keeps track of the monochromatic segments at each level, and only checks 4-cousins along segment boundaries. Since there are $O(k)$ segments at each level (this can be verified by induction) we need $O(k)$ space for processing the current level, and $O(N)$ additional space to record the conflicts. The running time is $O(k)$ per level to map the segments at the current level to new segments in the next level, and we iterate over $W-1$ levels.
\end{proof}

The following theorem provides an upper bound on the number of conflicts (size of the TCAM) generated by Algorithm~\ref{alg_consecutive_ranges}. Unfortunately, it is not much stronger than the naive upper-bound of $Wk$.

\begin{theorem}[Number of conflicts of Algorithm~\ref{alg_consecutive_ranges}]
\label{theorem_table_size_consecutive_ranges}
Let $N$ be the number of conflicts generated by Algorithm~\ref{alg_consecutive_ranges}. Then $N \le (W-\floor{\lg k} + 1) (k-1) + 1$.
\end{theorem}
\begin{proof}
A conflict occurs only when two siblings have different colors, otherwise their parent has the same color as both of them. Since the coloring begins as segments and continues this way, at most $k-1$ conflicts can happen at each of the first $W-\floor{\lg k}$ levels. The top $\floor{\lg k}$ levels of the trie contain at most $k$ nodes, so from level $\floor{\lg k}$ and upwards to the root (including the root) at most $k$ additional conflicts can be introduced. In total we get at most $(W-\floor{\lg k} + 1) (k-1) + 1$ conflicts. 
\end{proof}

\begin{theorem}[Algorithm~\ref{alg_consecutive_ranges} Optimality]
\label{theorem_segments_optimality_algorithm}
Algorithm \ref{alg_consecutive_ranges} generates a trie coloring $\chi^*$ that extends the input coloring $\chi$ with a minimum number of conflicts.
\end{theorem}

To prove Theorem~\ref{theorem_segments_optimality_algorithm} it suffices to argue why Algorithm~\ref{alg_consecutive_ranges} is a special case of a more general algorithm called ORTC \cite{DravesKSZ99} that works for any input coloring of the leaves, not necessarily of $k$ consecutive ranges. That is, ORTC works for an arbitrary mapping from addresses (leaves of the trie) to targets (the colors).

The ORTC algorithm can be thought of as having two main phases: a bottom-up phase in which each node receives information regarding its subtree to know its descendants' colors, and then a top-down phase which can be thought of as making the necessary tie-breaking by coloring of the parent. Our special case does not require the top down phase.

Concretely, ORTC starts at the leaves, and aggregates upward the candidate color sets $C_u$ of each node $u$ based on its children $v,w$ as follows. If $C_v \cap C_w = \emptyset$ then $C_u = C_v \cup C_w$, otherwise $C_u = C_v \cap C_w$. 
This bottom-up process starts by assigning $C_v$ to be a singleton set containing the color of $v$ for each leaf $v$. Then, a second pass on the trie from top to bottom picks for each node $v$ the color of its parent if this color is in $C_v$ and an arbitrary color from $C_v$ if it is not.

For example, in our simplified case, if we start with $4$ leaves whose color-sets are $\{a\},\{b\},\{b\},\{d\}$, then the color-sets of the parents will be $\{a,b\}$ and $\{b,d\}$ (union), respectively. The color-set of the grandparent will be $\{b\}$ (intersection), which will then propagate downwards in the second phase to color the parents with $b$, as in our algorithm. Our algorithm may produce a different output than ORTC since ORTC chooses an arbitrary color for a node if its list doesn't contain its parent color, while Algorithm~\ref{alg_consecutive_ranges} is stated more strictly. For example, in a trie with four leaves, each of a different color, ORTC can assign any color to the root, while Algorithm~\ref{alg_consecutive_ranges} will color it like the right-most leaf (see Fig.~\ref{figure_example_ortc_vs_alg}). This strictness makes our algorithm very concise, but in essence it is still a special-case of ORTC.

\begin{figure}[t]
  \centering
  \includegraphics[width=0.7\linewidth]{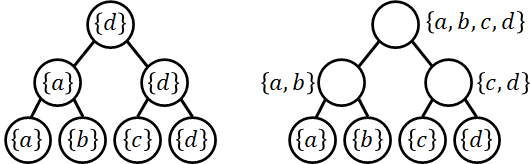}
  \caption{Coloring a toy example of a trie of depth $2$ with Algorithm~\ref{alg_consecutive_ranges} versus ORTC. Algorithm~\ref{alg_consecutive_ranges} has a single outcome (left), ORTC has several degrees of freedom (right): it can choose any of the four colors for the root, and then any of the two colors for the child that doesn't have this color in its list.}
  \label{figure_example_ortc_vs_alg}
\end{figure}

To elaborate how our algorithm is a special case, note that when $b=c$ (in the terms of Algorithm~\ref{alg_consecutive_ranges}), the color that ORTC will propagate to the grandparent must be $b$, because the parents will have the sets $\{a,b\}$ and $\{b,d\}$ whose intersection is non-empty. It also cannot be that $a=d$ is another candidate color for the grandparent, unless $a=d=b$ since the colors are segmented. Therefore, during the top-down phase the grandparent and the two parents will be colored by $b$ as stated by Algorithm~\ref{alg_consecutive_ranges}. On the other hand, if $b \ne c$, then the candidate colors of the parents will be $\{a,b\}$ and $\{c,d\}$, and the grandparent's will be $\{a,b,c,d\}$. During the top-down phase, ORTC allows to choose either $a$ or $d$ (or $b$ or $c$), and in accordance the color of the parents can be chosen to be $a$ and $d$, just as Algorithm~\ref{alg_consecutive_ranges} chooses.

\subsection{Computing the Best Segmented Allocation}

Given a partition $P = [p_1,\ldots,p_k]$ of $2^W$ into $k$ parts, Algorithm~\ref{alg_consecutive_ranges} does not directly tell us if $P$ has a smallest TCAM representation such that target $i$ gets $p_i$ consecutive addresses. However, if $k$ is not too large, we can still use Algorithm~\ref{alg_consecutive_ranges} to decide if there is an optimal allocation in segments. For each permutation $\pi$ of $p_1,\ldots,p_k$ we color the leaves of the trie consistently with $\pi$,\footnote{It is enough to check only $k!/2$ permutations if we exploit left to right symmetry, and even less if some parts of the partitions are equal.} and use Algorithm~\ref{alg_consecutive_ranges} to find a smallest TCAM representation that allocates the addresses in segments ordered by $\pi$. Then we also compute the complexity $\lambda(P)$ (recall Definition~\ref{definition_length_of_partition}). If there is a permutation $\pi^*$ for which the number of conflicts in its associated optimal trie coloring is $\lambda(P)$ then $P$ can be realized optimally in segments, and otherwise it cannot.

For specific fixed values of $k$ and $W$ there are $\binom{2^W-1}{k-1}$ ordered partitions of $2^W$ into $k$ parts. If we apply Algorithm~\ref{alg_consecutive_ranges} and compute the complexity of each of these partitions\footnote{It is enough to compute $\lambda(P)$ once for all permutations of the same unordered partition.} then we can conclude if there exists a partition of $2^W$ into $k$ parts that cannot be realized optimally in segments. We implemented this method and applied it to $k=3$ and $W \le 8$. The bound of $W \le 8$ was known to be sufficient since we knew that the partition $[12,49,195]$ cannot be realized optimally in segments by applying the techniques described in Section~\ref{section_segments_generalized}.

We found that $P^* = [13,13,6]$ is the only partition of $2^5$ into three parts that does not have an optimal representation that partitions the addresses into three segments, Theorem~\ref{theorem_minimal_example} states this. The uniqueness of $P^*$ for $W=5$ implies that there is no such partition $P$ for $k=3$ and $W' < 5$ (which we have seen by enumeration), since such a partition $P$ would give another partition that cannot be realized optimally in segments for $k=3$ and $W=5$ simply by multiplying each weight in $P$ by $2^{W-W'}$.

To figure out if there is a partition of $2^W$ for $W<5$ into $k>3$ parts that cannot be realized optimally in segments we made an exhaustive search for $3 \le k \le 2^W < 2^5$ (i.e.\ $W=2,3,4$), and did not find such a partition.\footnote{The complexity of an exhaustive search for fixed $W$ enumerates $\sum_{k=1}^{2^W}{\binom{2^W-1}{k-1}} = 2^{2^W-1} = O(2^{2^W})$ partitions. This complexity also subsumes an exhaustive search for all $W' \le W$. For $W \le 4$, the exhaustive search is feasible.} Hence, $[13,13,6]$ is truly minimal, simultaneously in $W$ and in $k$. 

\begin{theorem}[Partition with sub-optimal segments realization]
\label{theorem_minimal_example}
The partition $P^* = [13,13,6]$ ($k=3$,$W=5$) cannot be realized optimally as three segments, one for each target.
\end{theorem}

\begin{proof}
Fig.~\ref{figure_mininal_example_non_consecutive} shows a trie and its corresponding TCAM that are constructed by the Bit Matcher algorithm for $P^*$, compared to the two outputs of Algorithm~\ref{alg_consecutive_ranges} on the two different possible permutations $[13,13,6]$ and $[13,6,13]$ ($[6,13,13]$ is symmetric to $[13,13,6]$). One can verify that $\lambda(P^*) = 5$, while Algorithm~\ref{alg_consecutive_ranges} requires $6$ transactions for any ordering of the segments.
\end{proof}

\begin{figure}[t!]
    \centering
    \subfigure[Optimal TCAM produced by Bit Matcher ($5$ rules)]{
    \includegraphics[width=0.9\linewidth]{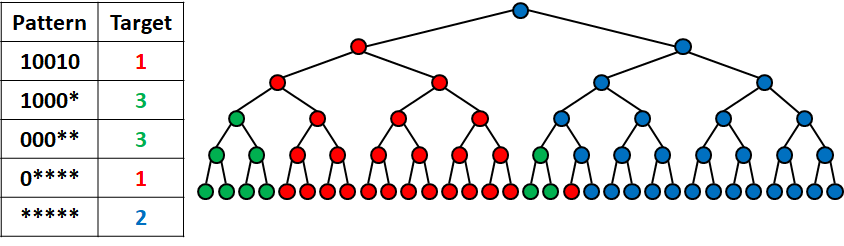}
    \label{figure_mininal_example_non_consecutive_bitmatcher_subfigure}
    }
    
    \subfigure[Optimal TCAM for ordered segments of lengths $13,13,6$ ($6$ rules)]{
    \includegraphics[width=0.9\linewidth]{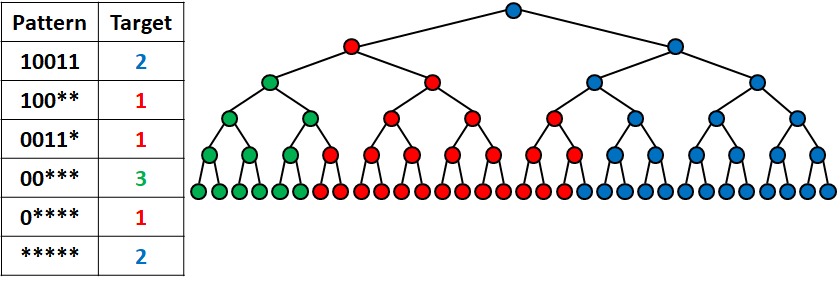}
    \label{figure_mininal_example_non_consecutive_13136_subfigure}
    }
    
    \subfigure[Optimal TCAM for ordered segments of lengths $13,6,13$ ($6$ rules)]{
    \includegraphics[width=0.9\linewidth]{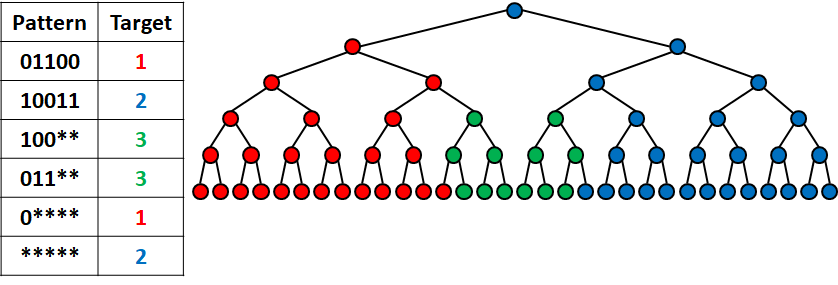}
    \label{figure_mininal_example_non_consecutive_13613_subfigure}
    }

  \caption{The partition $P^* = [13,13,6]$ cannot be realized optimally while also allocating a single segment of addresses per target. The cost of having segments is one additional rule in the TCAM table ($6$ rules, rather than $5$).}
  \label{figure_mininal_example_non_consecutive}
\end{figure}

\begin{remark}
Once we have a ``problematic'' partition, such as $P = [13,13,6]$, we can always ``embed'' it among partitions with larger sum and/or more parts.

This is trivially true of parts if size $0$ are allowed, since we can take $[p_1,\ldots,p_k]$ and construct from it $Q = [2^h \cdot p_1, \ldots, 2^h \cdot p_k, 0,\ldots,0]$ for $k\ \ge k$ and $W' = W+h$. If $0$ parts are not allowed, we can perturb $Q$ slightly, for example $Q = [2^h \cdot p_1, \ldots, 2^h \cdot p_k - 1, 1]$ if $k' = k+1$. Lemma~\ref{lemma_forced_transaction} (later) can be used to formally justify the fact that such small perturbations preserve the underlying ``problematic'' behavior.
\end{remark}

We conclude this section with an alternative proof to the fact that when $k=2$ every partition can be realized with the minimum number of rules while allocating a segment to each target (originally proven in \cite{AccurateExp}).

\begin{theorem}[Optimal allocation in segments for $k=2$]
\label{theorem_consecutive_k2}
Let $P = [p, 2^W-p]$ be any partition of $2^W$ into two parts. Then it has a smallest TCAM realization such that the first $p$ addresses are mapped to the first target, and the last $2^W-p$ addresses are mapped to the second target.
\end{theorem}

\begin{proof}
Since there are only two targets, Algorithm~\ref{alg_consecutive_ranges} produces at most one conflict per level. Furthermore, it indeed produces a conflict at level $\ell$ (leaves are at level $0$) if the transition between the two segments at that level is in $4$-cousins colored $a$,$a$,$a$,$b$ or $a$,$b$,$b$,$b$ from left to right (a $4$-cousin transition of the form $a,a,b,b$, or $a,a,a,a$, or $b,b,b,b$ does not produce a conflict). Algorithm~\ref{alg_consecutive_ranges} colors the parents $a,a$ in the first case ($a$,$a$,$a$,$b$), producing a transaction from $b$ to $a$, and colors the parents $b,b$ in the second case ($a$,$b$,$b$,$b$), producing a transaction from $a$ to $b$.

We prove by induction on the levels that Bit Matcher produces exactly the same transactions at levels $0,\ldots,\ell$. Assuming this is true up to level $\ell-1$, then the weights maintained by Bit Matcher after performing its transactions at levels $0,\ldots,\ell-1$ are exactly $x \cdot 2^\ell$ and $y \cdot 2^\ell$ where $x$ is the length of the first segment at level $\ell$ and $y$ is the length of the second segment at level $\ell$ as maintained by Algorithm~\ref{alg_consecutive_ranges} ($x+y = 2^{W-\ell}$). Bit Matcher performs a transaction between $x$ and $y$ iff both are odd, which happens iff the $4$-cousins where the transition between the segments occurs are $a$,$a$,$a$,$b$ or $a$,$b$,$b$,$b$. If $x$ equals $1$ modulus $4$ then $y$ equals $3$ modulus $4$ and the transaction is from $x$ to $y$ by bit-lexicographic order (defined in Algorithm~\ref{alg_bit_matcher}), and otherwise ($x = 3$ modulus $4$) it is from $y$ to $x$. This exactly corresponds to the transaction of the conflict produced by Algorithm~\ref{alg_consecutive_ranges}.
\end{proof}

\section{Generalized Segmentation}
\label{section_segments_generalized}

In the previous section we presented a partition that cannot be realized such that each target is allocated as a single segment. In this section we generalize this result to show that even if we allow each part to be fragmented into $M$ segments, there are still ``hard partitions'' which cannot be realized in an optimal TCAM that satisfies this constraint on the fragmentation. We achieve this by the following steps:
\begin{enumerate}[leftmargin=*,label=(\arabic*),noitemsep]
    \item In Subsection~\ref{subsection_transactions_graph} we define a transactions-graph (Definition~\ref{definition_graph_of_transactions}) of a sequence of transactions $s$. We can also think of the transactions-graph of a colored trie that implies a sequence $s$.
    
    \item In Subsection~\ref{subsection_graphs_of_segmened_coloring}, we prove that a coloring of a trie with $k$ colors (targets), must have a certain fragmentation in order for its transactions-graph to be a clique.
    
    \item In Subsection~\ref{subsection_graphs_of_forced_cliques}, we show how to construct a partition $P$ such that any shortest sequence that induces it has a transactions-graph which is a clique.
    
    \item Finally, Subsection~\ref{subsection_fragmentation_conclusion} combines the results to deduce that some partitions into $k$ parts must fragment at least one target into $\frac{k+1}{4} + \frac{1}{2k}$ parts, in order to be realized with the minimum possible number of TCAM rules.
\end{enumerate}

\subsection{Transactions-Graph}
\label{subsection_transactions_graph}

In this subsection we define the transactions-graph induced by a sequence of transactions.

\begin{definition}[Transactions-Graph]
\label{definition_graph_of_transactions}
Let $s$ be a sequence of transactions which contains only a single transaction with the target $0$. Let $P$ be the partition induced by $s$. Let $A_\ell$ denote the subset of all transactions of size no larger than $2^\ell$. Let $L$ be the largest integer such that when we apply all transactions of $A_L$ to $P$, all the weights remain positive. We define the \textbf{transactions-graph} of $s$ to be an undirected graph with $k$ nodes, one for each target (ignoring $0$). The graph has an edge $(i,j)$ if $\TRANSACT{i}{2^\ell}{j} \in A_L$ or $\TRANSACT{j}{2^\ell}{i} \in A_L$ for some $\ell \le L$. See Fig.~\ref{figure_transactions_graph} for example.
\end{definition}

\begin{figure}[ht]
  \centering
  \includegraphics[width=0.5\linewidth]{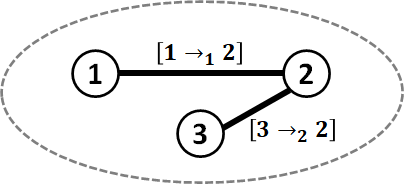}
  \caption{Illustrating a transactions-graph for the sequence $s = \TRANSACT{1}{1}{2}\TRANSACT{3}{2}{2}\TRANSACT{3}{4}{1}\TRANSACT{1}{16}{2}\TRANSACT{2}{32}{0}$. The sequence induces the partition $P = [13,13,6]$. For this sequence: $L=1$ and $A_L = \{\TRANSACT{1}{1}{2} , \TRANSACT{3}{2}{2} \}$.}
  \label{figure_transactions_graph}
\end{figure}

It is crucial to note that the transactions-graph is a property of the set of transactions in $A_L$, and it does not depend on their order in $s$. Moreover, it does not correspond to the partition since two different sequences that induce the same partition may have different transactions-graphs, and two different partitions may be induced by sequences with the same transactions-graphs.

\begin{remark}
\label{remark_graph_of_coloring}
A coloring of a trie $\chi$ corresponds to a unique set of transactions (the converse is not true, see Figure~\ref{figure_sequence_to_tries}), so the transactions-graph of a colored trie is the transactions-graph of its corresponding sequence.
\end{remark}

\begin{figure}[t]
  \centering
  \includegraphics[width=0.9\linewidth]{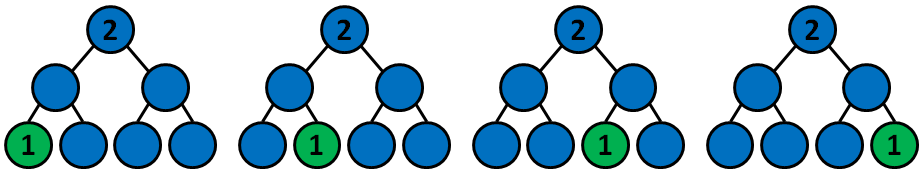}
  \caption{A sequence with more than a single transaction does not correspond to a unique coloring of the trie since we can swap subtrees whose roots are siblings. In this example, $s = \TRANSACT{1}{1}{2}\TRANSACT{2}{4}{0}$ and we have four different colored tries that correspond to this set of transactions.}
  \label{figure_sequence_to_tries}
\end{figure}

\subsection{Transactions-Graph of a Segmented Coloring}
\label{subsection_graphs_of_segmened_coloring}
In this subsection we analyze the structure of the transactions-graph of a sequence that corresponds to some colored trie. Concretely, let $\chi$ be a specific coloring of the leaves of the trie and let $\chi^*$ be an extension of $\chi$ with minimum conflicts. We claim that we can order the transactions corresponding to $\chi^*$ such that each transaction is between neighboring colors. We formalize this claim as follows.

\begin{definition}[Neighbouring leaves]
\label{definition_linear_neighbour}
Two leaves are neighbours if they are consecutive when we traverse the leaves from left to right.
Notice that two leaves which are siblings are also neighbors, but each leaf except the leftmost and the rightmost ones has another neighbor. For example, the leaves representing $01$ and $10$ on the trie of depth $2$ are also neighbours.
\end{definition}

\begin{definition}[Transaction of neighbouring colors]
\label{definition_neighbour_colors}
Assume that we apply the transactions in bottom-up order, with respect to depth in the trie. We say that a transaction is between neighbouring colors if when it is applied, it re-colors a leaf by the color of one of its neighbours.
\end{definition}

Now we prove that we can order any sequence so that its transactions are always of neighbouring colors. We break it down to the following lemmas.

\begin{lemma}
\label{lemma_path_root_to_leaf}
Let $\chi$ be a fixed coloring of the leaves of the trie, and let $\chi^*$ be an extension of $\chi$ to the whole trie, with a minimum number of conflicts. Then $\chi^*$ must contain a monochromatic path between the root and at least one of the leaves.
\end{lemma}

\begin{proof}
Let $S$ be the set of all nodes in the trie that have the same color as the root and are connected to the root by a monochromatic path, including the root itself. If $S$ does not contain a leaf, then there is at least one node $v$ (not necessarily a leaf) that doesn't belong to $S$ but its parent is in $S$. If we change the color of every node in $S$ to the color of $v$, we strictly reduce the number of conflicts in $\chi^*$, while keeping $\chi$ intact. This contradicts our assumption that $\chi^*$ has the smallest number of conflicts among all extensions of $\chi$. Hence, $S$ contains a leaf, and the root has a monochromatic path to this leaf.
\end{proof}

\begin{lemma}
\label{lemma_path_hanging_subtrees_with_minimum_conflicts}
Let $\chi$ be a fixed coloring of the leaves of the trie, and let $\chi^*$ be an extension of $\chi$ to the whole trie, with a minimum number of conflicts. Let $v$ be a node in conflict with its parent in $\chi^*$ and let $T_v$ be the subtree rooted at $v$. Then the restriction of $\chi^*$ to $T_v$ is a coloring of $T_v$ with minimum conflicts that extends the restriction of $\chi$ to $T_v$.
\end{lemma}

\begin{proof}
Assume by contradiction that the claim is not true. This means that we can reduce the number of conflicts in $T_v$ by recoloring it. Since $v$ is already in conflict with its parent in $\chi^*$, this recoloring must also reduce the number of conflicts in $\chi^*$. But this contradicts our assumption that $\chi^*$ has the smallest number of conflicts among all colorings consistent with $\chi$.
\end{proof}

\begin{lemma}
\label{lemma_transactions_ordered_by_neighbours}
Let $\chi$ be a fixed coloring of the leaves of the trie, and let $\chi^*$ be an extension of $\chi$ to the whole trie, with a minimum number of conflicts. Consider the transactions corresponding to the conflicts of $\chi^*$ excluding the transaction into target $0$ which correspond to the conflict at the root. It is possible to order the transactions in bottom-up order such that when applied in order, each transaction is a transaction of neighbouring colors when it is applied.
\end{lemma}

\begin{proof}
We prove the claim by induction on the depth of the trie, see also Figure~\ref{figure_order_for_coloring_by_neihgbours} to clarify the steps. The basis of the induction is a trie of depth $0$ for which the claim holds vacuously.

For the induction step assume that the claim is true for any trie of depth at most $D$, and consider a trie of depth $D+1$. Let $S$ be the set of all nodes connected to the root by a monochromatic path. Each node $v \notin S$ such that its parent is in $S$ defines a subtree $T_v$ of depth at most $D$. By Lemma~\ref{lemma_path_hanging_subtrees_with_minimum_conflicts}, the restriction of $\chi^*$ to $T_v$ is of minimum number of conflicts, so by the induction assumption we can order the transactions that correspond to the conflicts in $T_v$, such that all of them are between neighbouring colors (when applied in order). We merge all the transactions of these subtrees such that the relative order of transactions when restricted to a subtree maintained in bottom-up order, the simplest way is to just concatenate the sub-sequences one after the other.

The coloring of the trie after applying these transactions consists of a connected monochromatic set of nodes $S$ that includes the root, from which monochromatic subtrees are hanging, none of them has the same color as the root. Now, by Lemma~\ref{lemma_path_root_to_leaf} we know that $S$ contains at least one leaf. Let $v$ be an arbitrary leaf in $S$. We scan the leaves starting from $v$ and spreading outwards, for instance all the way left, and then all the way right. Whenever we encounter a leaf $u$ whose color is different than the color of $S$, we apply the transaction that corresponds to the root of the monochromatic tree that contains $u$. This sequentially recolors the leaves by the color of $v$ and therefore by definition each of the transactions that we apply is between neighbouring colors, see also Figure~\ref{figure_order_for_coloring_by_neihgbours}. Since we recolor disjoint subtrees it is clear that we get a bottom-up sequence.
\end{proof}

\begin{figure}[t!]
    \centering
    \subfigure[Input coloring $\chi$ (leaves) extended to $\chi^*$]{
    \includegraphics[width=0.7\linewidth]{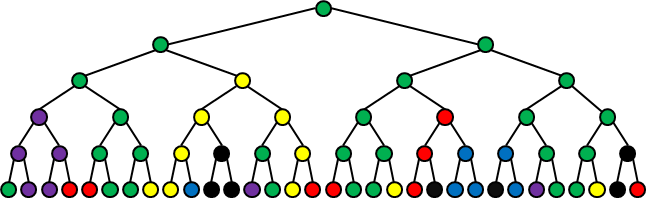}
    \label{figure_order_for_coloring_by_neihgbours_initial}
    }
    
    \subfigure[State of coloring after transactions are applied to solve smaller subtrees, such that they become monochromatic. The order of re-coloring the subtrees emanates from the leaf marked in a square.]{
    \includegraphics[width=0.7\linewidth]{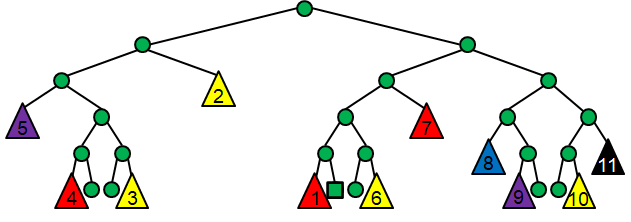}
    \label{figure_order_for_coloring_by_neihgbours_hanging_subtrees}
    }

  \caption{An illustration of how Lemma~\ref{lemma_transactions_ordered_by_neighbours} recursively orders the transactions corresponding to a coloring with minimum number of conflicts. \ref{figure_order_for_coloring_by_neihgbours_initial} shows the original trie, colored with six colors. The set $S$ contains all the green nodes which are connected by green path to the root ($S$ does not contain all green nodes). \ref{figure_order_for_coloring_by_neihgbours_hanging_subtrees} shows the inductive step: let $s_i$ be the sequence of transaction that makes the $i$th subtree monochromatic such that each is a transaction between neighbouring colors. Then the final sequence of transaction begins with a concatenation of $s_1,\ldots,s_{11}$ in arbitrary order. It ends with transactions that re-color each subtree hanging off of $S$ to green, in the order marked in the figure, emanating away from the green square leaf.}
  \label{figure_order_for_coloring_by_neihgbours}
\end{figure}

\begin{lemma}
\label{lemma_graph_too_sparse_for_segments}
Let $\chi$ be a fixed coloring of the leaves of the trie with $k$ colors, in $d$ monochromatic segments. Let $\chi^*$ be an extension of $\chi$ to the whole trie, with a minimum number of conflicts. Then the transactions-graph of the sequence corresponding to $\chi^*$ contains at most $2d - k - 1$ edges.
\end{lemma}
\begin{proof}
By definition, the transactions-graph of a sequence does not depend on the order of the transactions. Moreover, recall that there is an edge between colors $c_1$ and $c_2$ if there is at least one transaction between them (there could be many such transactions).

By Lemma~\ref{lemma_transactions_ordered_by_neighbours} we can order the transactions bottom-up such that transactions are always between neighbouring colors, i.e.\ between colors that are associated with adjacent segments. Therefore by the preceding discussion we have to bound the number of pairs of colors with adjacent segments. We bound this by bounding the total number of adjacent segments of different colors as follows.

\begin{enumerate}[leftmargin=*,label=(\arabic*),noitemsep]
    \item Initially, there are $d$ monochromatic segments overall, which yields $d-1$ pairs of adjacent segments of different colors.
    
    \item Two segments that are not adjacent initially may become adjacent. For example, in Figure~\ref{figure_order_for_coloring_by_neihgbours_initial} 
    no red and blue segments are adjacent, but after recoloring the monochromatic subrees, we get such an adjacent pair in Figure~\ref{figure_order_for_coloring_by_neihgbours_hanging_subtrees} (subtrees 7 and 8). Two segments become adjacent when all segments between them have been swallowed, due to re-coloring as transactions are applied. It follows that the number of pairs of segments that become adjacent is bounded by the number of segments that are swallowed, which is at most $d-k$, since by definition a single segment of each color must remain as long as we still produce edges of the transactions-graph.
\end{enumerate}

Overall, we get at most $(d-1) + (d-k) = 2d-k-1$ pairs of adjacent segments, which is an upper bound on the number of edges in the transactions-graph of $\chi^*$.
\end{proof}

\begin{corollary}
\label{corollary_minimum_clique_fragmentation}
Let $\chi^*$ be a coloring of a trie, with $k$ colors and a minimum number of conflicts. If the transactions-graph of the sequence corresponding to $\chi^*$ is a clique, then in the restriction of $\chi^*$ to the leaves of the trie, at least one of the colors is fragmented into at least $\frac{k+1}{4} + \frac{1}{2k}$ segments.
\end{corollary}

\begin{proof}
A clique over $k$ nodes has $\binom{k}{2}$ edges. Let $d$ be the number of monochromatic segments when we restrict $\chi^*$ to the leaves. By Lemma~\ref{lemma_graph_too_sparse_for_segments} we have $\binom{k}{2} \le 2d - k - 1 \Rightarrow \frac{k^2 + k}{4} + \frac{1}{2} \le d$. It follows that at least one color must have at least $\frac{d}{k} \ge \frac{k+1}{4} + \frac{1}{2k}$ segments.
\end{proof}

\subsection{Forcing Transactions}
\label{subsection_graphs_of_forced_cliques}
Now we take a small detour from trie coloring, and go back to deal purely with transactions and sequences. We present a method to construct a partition, such that any shortest sequence that induces it contains a \textbf{fixed} set of transactions.

\begin{lemma}[Forced Transaction]
\label{lemma_forced_transaction}
Let $P = [p_1, \ldots, p_k]$ be any partition such that $\forall i: p_i > 0$, $k \ge 2$ and $\sum_i p_i = 2^W$. Let $Q = [q_1, \ldots, q_k]$ be such that: $q_1 = 8 p_1 - 1$, $q_2 = 8p_2 + 1$, $\forall i \ge 3: q_i = 8 p_i$.
Then $Q$ is a partition such that $\forall i: q_i > 0$, $\sum_i q_i = 2^{W+3}$, and any shortest sequence that induces $Q$, contains the transaction \TRANSACT{2}{1}{1}.
\end{lemma}

\begin{proof}
First we prove the basic claims on $Q$: Since $\forall i: p_i > 0 \Rightarrow 8p_i \ge 8 \Rightarrow q_i \ge 7 > 0$. Also, $\sum_i q_i = 8 \sum_i p_i = 8 \cdot 2^W$. It remains to show that every shortest sequence that induces $Q$ contains the transaction \TRANSACT{2}{1}{1}.

Let $s^L$ be the subset of $s$ consisting of all transactions of sizes at most $4$, and let $s^H$ be the subset of $s$ consisting of all transactions of sizes larger than $4$. We break every transaction $\TRANSACT{a}{2^\ell}{b}$ in $s^L$ into two transactions $\TRANSACT{a}{2^\ell}{\star}\TRANSACT{\star}{2^\ell}{b}$ for some artificial target denoted by $\star$. Denote this new sequence by $s^L_\star$. Clearly, $|s^L_\star| = 2|s^L|$. We also order in $s^L_\star$ the transactions in which $\star$ is a receiver before the transactions in which it is a sender. This means that $\star$ aggregates weight from the senders, and then spreads it to the receivers. Since the size of every transaction in $s^H$ is a multiple of $8$, it means that applying $s^L$ to $Q$ makes all the weights divisible by $8$. Therefore, $\star$ aggregates $7$ modulo $8$ from the first target, $1$ modulo $8$ from the second target, and $0$ modulo $8$ from any other target. Spreading this weight is done in multiples of $8$ to any receiver.

We now modify $s^L_\star$ to make it shorter and then eliminate the artificial target. The resulting sequence $s'$ will be strictly shorter than $s^L$ unless $s^L$ contains $\TRANSACT{2}{1}{1}$, and the concatenation of $s'$ and $s^H$ will still induce $Q$. Since we assumed that $s$ is a shortest sequence inducing $Q$, this implies that $s$ must contain $\TRANSACT{2}{1}{1}$. We define the \emph{net weight} delivered from a target $j$ to $\star$ as the sum of the sizes of the transactions in which $j$ is a sender, minus the sum of the sizes of the transactions in which $j$ is a receiver. We modify $s^L_\star$ as follows:
\begin{enumerate}[leftmargin=*,label=(\arabic*),noitemsep]
    \item For $j \ne 1,2$, the net weight delivered is $8m$:
    \begin{enumerate}[leftmargin=*,label=(\alph*),noitemsep]
        \item $m>0$: we replace in $s^L_\star$ all the transactions of the forms $\TRANSACT{j}{2^\ell}{\star}$ and $\TRANSACT{\star}{2^\ell}{j}$ by $m$ transactions $\TRANSACT{j}{8}{\star}$, which is shorter.
        \item $m<0$: we replace in $s^L_\star$ all the transactions of the forms $\TRANSACT{j}{2^\ell}{\star}$ and $\TRANSACT{\star}{2^\ell}{j}$ by $m$ transactions $\TRANSACT{\star}{8}{j}$, which is shorter.
        \item $m=0$: we delete all the transactions of the forms $\TRANSACT{j}{2^\ell}{\star}$ and $\TRANSACT{\star}{2^\ell}{j}$. We reduced the number of transactions if $j$ participates in $s^L_\star$.
    \end{enumerate}
    
    \item For $j = 2$, the net weight delivered is $8m + 1$:
    \begin{enumerate}[leftmargin=*,label=(\alph*),noitemsep]
        \item $m>0$: we replace in $s^L_\star$ all the transactions of the forms $\TRANSACT{j}{2^\ell}{\star}$ and $\TRANSACT{\star}{2^\ell}{j}$ by $m$ transactions $\TRANSACT{j}{8}{\star}$ and a transaction $\TRANSACT{j}{1}{\star}$, which is shorter.
        \item $m<0$: we replace in $s^L_\star$ all the transactions of the forms $\TRANSACT{j}{2^\ell}{\star}$ and $\TRANSACT{\star}{2^\ell}{j}$ by $m$ transactions $\TRANSACT{\star}{8}{j}$ and a transaction $\TRANSACT{j}{1}{\star}$, which is shorter.
        \item $m=0$: we replace in $s^L_\star$ all the transactions of the forms $\TRANSACT{j}{2^\ell}{\star}$ and $\TRANSACT{\star}{2^\ell}{j}$ by the transaction $\TRANSACT{j}{1}{\star}$. We reduced the number of transactions unless this was the only transaction of $j$ in $s^L$ to begin with.
    \end{enumerate}
    
    \item For $j = 1$, the net weight delivered is $8m + 7$:
    \begin{enumerate}[leftmargin=*,label=(\alph*),noitemsep]
        \item $m > -1$: we replace in $s^L_\star$ all the transactions of the forms $\TRANSACT{j}{2^\ell}{\star}$ and $\TRANSACT{\star}{2^\ell}{j}$ by $m+1$ transactions $\TRANSACT{j}{8}{\star}$ and an additional transaction $\TRANSACT{\star}{1}{j}$, which is shorter (sending $7$ explicitly by transactions of size at most $4$ requires three transactions, compared to only two transactions by sending $8$ and receiving $1$).
        \item $m < -1$: we replace in $s^L_\star$ all the transactions of the forms $\TRANSACT{j}{2^\ell}{\star}$ and $\TRANSACT{\star}{2^\ell}{j}$ by $m+1$ transactions $\TRANSACT{\star}{8}{j}$ and an additional transaction $\TRANSACT{\star}{1}{j}$, which is shorter.
        \item $m=-1$: we replace in $s^L_\star$ all the transactions of the forms $\TRANSACT{j}{2^\ell}{\star}$ and $\TRANSACT{\star}{2^\ell}{j}$ by the transaction $\TRANSACT{\star}{1}{j}$. We reduced the number of transactions unless this was the only transaction of $j$ in $s^L$ to begin with.
    \end{enumerate}
\end{enumerate}

The modified sequence contains the transactions $\TRANSACT{2}{1}{\star}$ and $\TRANSACT{\star}{1}{1}$ and the rest of the $\star$ transactions are of size $8$. Therefore we can get rid of $\star$ by repetitively replacing pairs of transactions of the form $\TRANSACT{a}{2^\ell}{\star}\TRANSACT{\star}{2^\ell}{b}$ by $\TRANSACT{a}{2^\ell}{b}$. Let $s'$ be the resulting sequence. Observe that $s'$ contains $\TRANSACT{2}{1}{1}$ and that $|s'| \le |s^L|$ (because $2|s'| \le 2|s^L_\star|$). The concatenation of $s'$ and $s^H$ induces $Q$. Further more, $s'$ is strictly shorter than $s$ unless the only transactions in $s^L_\star$ are $\TRANSACT{2}{1}{\star}\TRANSACT{\star}{1}{1}$.
\end{proof}

\begin{remark}
\label{remark_forced_by_factor_eight}
Note that $8$ is the minimal factor by which we have to scale $P$ for the proof of Lemma~\ref{lemma_forced_transaction} to work. If we tried using $4$ (such that the extra values would be 3 and 1 modulo 4), then it could be that $Q$ is induced by a shortest sequence $s$ that contains $\TRANSACT{2}{1}{1}\TRANSACT{1}{4}{2}$, in which case we can replace these two transactions by \TRANSACT{1}{1}{2}\TRANSACT{1}{2}{2}, which shows that \TRANSACT{2}{1}{1} is not forced. As a concrete example, consider the partition $P = [5,3]$. It can be induced with three transactions either by $\TRANSACT{2}{1}{1}\TRANSACT{2}{2}{1}\TRANSACT{1}{8}{0}$ or by $\TRANSACT{1}{1}{2}\TRANSACT{2}{4}{1}\TRANSACT{1}{8}{0}$.
\end{remark}

Now we extend this method of forcing a single transaction to force many.

\begin{theorem}[Forced Sequence]
\label{theorem_forced_unique_sequence_contruction}
Let $A = [(i_1,j_1),(i_2,j_2),\ldots,(i_N,j_N)]$ be a list of pairs of indices, between $1$ and $k$, and let $P^0$ be an arbitrary partition of $2^{W_0}$ into $k$ positive parts. Then there exists a partition $P$ of $2^{W_0 + 3N}$ into $k$ positive parts such that any shortest sequence that induces $P$ contains the transactions $\TRANSACT{i_m}{2^{3(m-1)}}{j_m}$ for $m=1,\ldots,N$, and after these $N$ transactions are applied the partition becomes $2^{3N} \cdot P^0$.
\end{theorem}

\begin{proof}
We construct $P$ by induction on $N$. The base case is for $N=0$: Since $A = \emptyset$, $P = P^0$ and the claim is trivially true.

For $N > 0$, let $P^{N-1}$ be a partition of $2^{W_0 + 3(N-1)}$ constructed for $P^0$ and the last $N-1$ pairs in $A$ by the induction hypothesis. We define $P^{N}$ as follows: $P^{N}_{i_1} = 8 P^{N-1}_{i_1}+1$, $P^{N}_{j_1} = 8 P^{N-1}_{j_1}-1$ and $\forall i \ne i_1,j_1: P^{N}_{i} = 8 P^{N-1}_{i}$. By Lemma~\ref{lemma_forced_transaction}: $P^N$ is a partition into positive parts, and $\TRANSACT{i_1}{1}{j_1}$ is forced in any shortest sequence that induces $P^N$. Consider a shortest sequence $s$ that induces $P^N$. After performing $\TRANSACT{i_1}{1}{j_1}$, $P^N$ becomes $8 \cdot P^{N-1}$. It follows that if we omit $\TRANSACT{i_1}{1}{j_1}$ from $s$ and divide the sizes of the transactions by $8$ we obtain a shortest sequence $s'$ that induces $P^{N-1}$ and vice versa. By induction, $s'$ must contain $\TRANSACT{i_m}{2^{3(m-2)}}{j_m}$ for $2 \le m \le N$.\footnote{Notice that the sizes of the transactions are $2^{3(m-2)}$ rather than $2^{3(m-1)}$ since for $s'$ the first pair is $(i_2,j_2)$, and we divided by $8$.} Therefore $s$ must contain $\TRANSACT{i_m}{2^{3(m-1)}}{j_m}$ for $2 \le m \le N$, and it also contains $\TRANSACT{i_m}{2^{3(m-1)}}{j_m}$ for $m=1$. Finally, since after performing all transactions in $s'$ between $i_m$ and $j_m$ for $2 \le m \le N$ we get $2^{3(N-1)} \cdot P^0$, it follows that after performing all transactions in $s$ corresponding to pairs in $A$ we get $2^{3N} \cdot P^0$.
\end{proof}

\begin{remark}
\label{remark_force_pairs_from_sequence}
Note that Theorem~\ref{theorem_forced_unique_sequence_contruction} gets as ``input'' pairs of indices. This means that we can take any sequence, in any order, consider its transactions as pairs, and then construct a new partition such that every shortest sequences that induces it must have transactions between these specific pairs. Loosely speaking, if the transactions are ordered in increasing sizes, we take the partition induced by the sequence and ``stretch'' its bits such that we get a partition that must have these particular transactions in any shortest sequence.
\end{remark}

\begin{corollary}
\label{corollary_induced_clique_graph}
There exists a partition $P$ such that every shortest sequence that induces $P$ has a transactions-graph which is a clique.
\end{corollary}
\begin{proof}
Let $A = \{(i,j) | 1 \le j < i \le k \}$, and let $P^0$ be an arbitrary partition into $k$ positive parts. By Theorem~\ref{theorem_forced_unique_sequence_contruction} there exists a partition $P$ such that every shortest sequence that induces $P$
has a transaction of size smaller than $2^{3|A|}$ between every $(i,j) \in A$. Furthermore, after we apply these transactions $P$ becomes $2^{3|A|} \cdot P^0$. This means that the transactions-graph of each such sequence is a clique. Different choices of $P^0$ yield different partitions that satisfy the claim.
\end{proof}

\subsection{Concluding The Proof}
\label{subsection_fragmentation_conclusion}

\begin{definition}[TCAM Fragmentation]
For any partition $P$ with $k$ parts, let $T$ be a TCAM that induces it. Let $m_i(T)$ denote the number of segments allocated to $p_i$ by this TCAM, and let $M(T) = max_{i\in[k]}{m_i}$. In words, $M(T)$ is the maximum fragmentation of a target.
\end{definition}

\begin{theorem}[Fragmentation versus TCAM Size]
\label{theorem_minimum_fragmentation_bound}
There exists a partition $P$ such that for any TCAM $T$ that induces $P$, either $M(T) \ge \frac{k+1}{4} + \frac{1}{2k}$ or $|T| > \lambda(P)$. In words: either the fragmentation of $P$ is large, or the number of rules is not optimal.
\end{theorem}

\begin{proof}
By Corollary~\ref{corollary_induced_clique_graph} there exists a partition $P$ such that the transactions-graph of any shortest sequence that induces $P$ is a clique. By Corollary~\ref{corollary_minimum_clique_fragmentation} it follows that a smallest TCAM $T^*$ for $P$ has at least one color that is fragmented to at least $\frac{k+1}{4} + \frac{1}{2k}$ segments.
Therefore, either we realize $P$ with $|T^*| = \lambda(P)$ rules and fragmentation of at least $\frac{k+1}{4} + \frac{1}{2k}$, or we use a TCAM $T$ with smaller fragmentation but then $|T| > \lambda(P)$.

Note that the statement for $k=1,2$ is trivial, because then $\frac{k+1}{4} + \frac{1}{2k} = 1$ and therefore the part $M(T) \ge 1$ in the statement is always satisfied.
\end{proof}

To make Theorem~\ref{theorem_minimum_fragmentation_bound} more concrete, observe that it generalizes Theorem~\ref{theorem_minimal_example}. In this case, the desired fragmentation is $1$, but for $k\ge 3$ we have $\frac{k+1}{4} + \frac{1}{2k} > 1$, so there exists a partition into $k=3$ parts such that in every smallest TCAM representing it, some target is fragmented. Theorem~\ref{theorem_minimum_fragmentation_bound} implies that an unavoidable maximum fragmentation increase by $1$ for every four additional parts in the (worst-case) partition.

Note that while the proof deals with a clique in the transactions-graph of every shortest sequence, this condition is sufficient but not necessary. For example, Figure~\ref{figure_transactions_graph} shows the transactions-graph of the minimal example $P = [13,13,6]$, in which we do not have a cycle (a $3$-clique). The third transaction of the sequence would have closed the cycle, but it zeroes one of the weights, and therefore does not contribute an edge to the transactions-graph.

We remark that the construction in Theorem~\ref{theorem_forced_unique_sequence_contruction} requires $3$ bit-levels per forced transaction. Therefore to have $\binom{k}{2}$ forced transactions the total width would be $W = 3\binom{k}{2} + \ceil{\lg k}$. The additional $\ceil{\lg k}$ levels are because all the weights are still non-zero after performing the transactions of the lowest $3\binom{k}{2}$ levels. While such a large $W$ is unrealistic for practical scenarios, this proves our theorem.

We can use fewer levels, by utilizing a level to induce several edges in the transactions-graph simultaneously. For example, if we begin with all weights odd, then level-$0$ induces $k/2$ edges at once. It is also possible to force a transaction sometimes with only 2 bit-levels as we show below. By using these optimizations it may be possible to reduce the width $W$ of the construction, but it should be done carefully without introducing unintended freedom that could break the enforcement of specific edges in the transactions-graph. In addition, the fact that a clique in the transactions-graph may not be necessary can also help reducing the number of levels.

As an example where shaving-off levels is possible, consider the partition $P = [12,49,195]$, whose Bit Matcher sequence is $s = \TRANSACT{2}{1}{3}\TRANSACT{3}{4}{1}\TRANSACT{1}{16}{2}\TRANSACT{2}{64}{3}\TRANSACT{3}{256}{0}$. One can verify that it uses only two levels per transaction, and that it cannot be realized optimally with a single segment per target as it requires six transactions for that, compared to only five transactions by Bit Matcher. The transactions here are in fact forced, even though by Remark~\ref{remark_forced_by_factor_eight} forcing using only two levels per transaction is generally not guaranteed.

\section{Looking for the Best Order of Segments}
\label{section_looking_for_best_order_of_segments}

In this section we revisit the simpler case of partitioning the address space with a single consecutive segment of addresses per target. Algorithm~\ref{alg_consecutive_ranges} assumes that the order of the segments is known, however it doesn't help us figure out this order. A natural question to ask is whether we can determine this order. In this sub-section we provide some partial guarantees that show that an arbitrary order is a $\min(k-1,W-\floor{\lg k} + 1)$ approximation compared to the best ordering and that by choosing a non-arbitrary ordering the gap may be slightly improved to $\min(\frac{k+1}{3},W-\floor{\lg k} + 1)$ (Theorem~\ref{theorem_approximation_segments_order}). We also demonstrate a partition in which the gap between the best and worst orderings is large (Theorem~\ref{theorem_approximation_segments_order_lower_bound_k_or_w}). Finally, we describe a seemingly natural greedy algorithm and explain why it may not work well. We emphasize that this section assumes that the desired partition $P$ to $k$ part must be realized as $k$ segments, one per part, at the cost of possibly more than $\lambda(P)$ rules.

\begin{definition}
Let $P$ be a partition of $2^W$ into $k$ parts, and let $\sigma \in S_k$ denote the permutation on the order of the parts. We denote by $N(P,\sigma)$ number of conflicts in the coloring produced by Algorithm~\ref{alg_consecutive_ranges} when run on $P$ in the order of $\sigma$.
\end{definition}

We begin with a ``warm-up'' approximation bound.
\begin{lemma}
\label{lemma_any_order_w_approximation}
Let $P$ be a partition of $2^W$ into $k$ parts, and let $\sigma^* \in S_k$ denote a permutation that minimizes $N(P,\sigma)$. Then for any $\sigma \in S_k$: $N(P,\sigma) < (W-\floor{\lg k} + 1) \cdot N(P,\sigma^*)$.
\end{lemma}

\begin{proof}
By Theorem~\ref{theorem_table_size_consecutive_ranges}, $N(P,\sigma) \le (W-\floor{\lg k} + 1) (k-1) + 1$, and we can relax this to $N(P,\sigma) < (W-\floor{\lg k} + 1) k$. Note also that $k \le \lambda(P) \le N(P,\sigma^*)$ since each color must have at least one conflict with some parent or as the root. Putting both together, we get $N(P,\sigma) < (W-\floor{\lg k} + 1) \cdot N(P,\sigma^*)$.
\end{proof}

The following approximation bound is less trivial.

\begin{lemma}
\label{lemma_any_order_k_approximation}
Let $P$ be a partition of $2^W$ into $k$ parts and let $\sigma^* \in S_k$ denote a permutation that minimizes $N(P,\sigma)$. Then for any $\sigma \in S_k$: $N(P,\sigma) \le N(P,\sigma^*) \cdot (k-1) - \frac{(k-2)(k+1)}{2}$.
\end{lemma}

\begin{figure}[t!]
    \centering
    \subfigure[Before new level: all in order]{
    \includegraphics[width=0.3\linewidth]{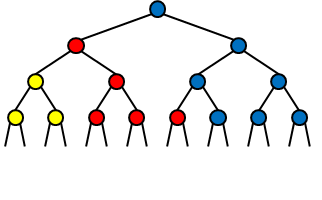}
    \label{figure_shift_subtrees_example_pre}
    }
    ~
    \subfigure[Begin new level: initial state]{
    \includegraphics[width=0.3\linewidth]{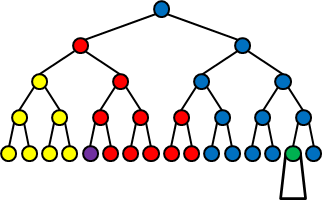}
    \label{figure_shift_subtrees_example_0}
    }
    ~
    \subfigure[Green shift 1]{
    \includegraphics[width=0.3\linewidth]{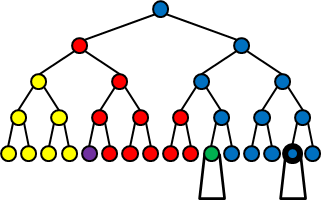}
    \label{figure_shift_subtrees_example_1}
    }
    ~
    \subfigure[Green shift 2]{
    \includegraphics[width=0.3\linewidth]{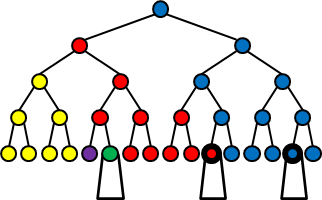}
    \label{figure_shift_subtrees_example_2}
    }
    ~
    \subfigure[Green shift 3]{
    \includegraphics[width=0.3\linewidth]{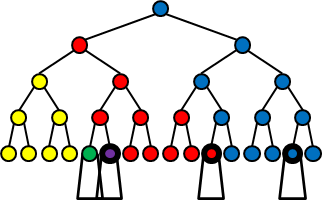}
    \label{figure_shift_subtrees_example_3}
    }
    ~
    \subfigure[Green shift 4]{
    \includegraphics[width=0.3\linewidth]{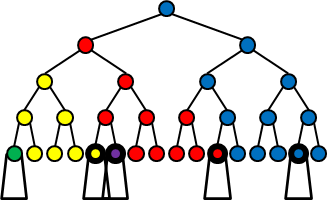}
    \label{figure_shift_subtrees_example_4}
    }

  \caption{Fixing the order of colors as explained in Lemma~\ref{lemma_any_order_k_approximation}. Assume that the desired order of the colors is [green, yellow, red, purple, blue], and that the first three levels (top-down) have been organized correctly [yellow, red, blue] (green and purple are missing), in \ref{figure_shift_subtrees_example_pre}. The initial state of the next level is given in \ref{figure_shift_subtrees_example_0}, and new green and purple nodes are out-of-place. The next steps \ref{figure_shift_subtrees_example_1}-\ref{figure_shift_subtrees_example_4} demonstrate how we fix the location of the green subtree. Doing so shifts $4$ other subtrees in the process, all of them are marked in bolded outline and a trapezoid. Note that on \ref{figure_shift_subtrees_example_1} no extra conflict is generated since the green node is moved (with its subtree) inside its parent color, and that on \ref{figure_shift_subtrees_example_3} no new conflict is generated because the purple node is itself in conflict. After the green node is set into its place, another shifting process can be applied to relocate the purple node.}
  \label{figure_shift_subtrees_example}
\end{figure}

\begin{proof}
Let $\sigma$ be some permutation and fix the coloring of the leaves according to it. We describe a way to color the trie according to $\sigma$, such that its number of conflicts $n$ satisfies $n \le (\lambda(P)-\frac{k}{2}) \cdot (k-1)$. By Theorem~\ref{theorem_segments_optimality_algorithm} we have $N(P,\sigma) \le n$, and by also substituting $\lambda(P) \le N(P,\sigma^*)$, the claim follows up to the additive term, which can be refined. In the rest of the proof, we describe how to achieve this coloring. See Figure~\ref{figure_shift_subtrees_example} for a concrete example of the arugments below.

As a first step, we initiate a coloring with $\lambda(P)$ conflicts by running Bit Matcher, or Niagara, or choosing any optimal coloring that ignores the segmentation requirement. By definition this coloring has the minimum number of conflicts possible for $P$, but its problem is that the leaves may not be colored in consecutive segments. To solve this problem, we work top-down from the root, and at every step re-order the subtrees to maintain the order required by $\sigma$. Initially, the root is ordered because it has a single node. When we go down, nodes with colors out of order may appear, and we will move around subtrees to fix the order. Note the following important facts:
\begin{enumerate}[label=(\arabic*),noitemsep]
    \item \label{fact_whole_subtrees} Moving a whole subtree doesn't affect the conflicts inside it. Therefore, moving a subtree rooted at a node $v$ may produce at most one additional conflict, in $v$.
    \item \label{fact_inside_parent} If $v$ is already in conflict, and we swap it with $u$ such that $u$ is of the same color of $v$'s parent, then no new conflict is created.
    \item \label{fact_conflicts_swap} If both $v$ and $u$ are already in conflict and in the same level, we don't increase the number of conflicts by swapping them (with their subtrees). Since the initial coloring had a minimum number of conflicts, we also don't save a conflict by such exchange.
\end{enumerate}

As long as we expand and see no conflict, we simply continue. At some point we may encounter a conflict, and we need to make sure that the color that appeared is in the right place. First, consider the case of a single conflict at the whole new level. Then we can shift the stray node $v$ to its place by moving at most $k$ nodes: $v$ itself, and the end-most node of each segment of color until $v$ reaches its desired location (it doesn't matter if it merges with an existing segment of the same color, or initiates this segment). The first shift of $v$ to the end of the segment of its parent color doesn't create a new conflict (by fact~\ref{fact_inside_parent}, see also Figure~\ref{figure_shift_subtrees_example_1}), and then it might be shifted over at most $k-2$ other segments, overall creating at most $k-2$ additional conflicts. Therefore, overall we may have at most $k-1$ conflicts compared to $1$ original conflict, See Figure~\ref{figure_shift_subtrees_example}.

Clearly, there could be several conflicts in the same level, say $m$. In this case, we need to show that moving them to the correct location doesn't result in more than $m(k-1)$ conflicts in the modified coloring, in order to maintain the $k-1$ ratio. This is not immediately trivial because the shifting argument of the previous paragraph assumes that no more than $k-1$ segments participate in the shift, but the new conflicting nodes are themselves segments of length $1$, so we need to be careful. However, by fact~\ref{fact_conflicts_swap}, shifting one conflict over another doesn't produce a new conflict (see Figure~\ref{figure_shift_subtrees_example_3}). Therefore the argument remains valid: fixing a specific conflict by shifting it to its correct location adds at most $k-2$ conflicts, and hence overall we get at most $m(k-1)$ conflicts in this level in the modified coloring, compared to $m$ in the original coloring (at the same level).\footnote{In fact, since exchanging two conflicting nodes can't add a new conflict, if there is more than one new conflict in the level, we can choose the order of the conflicts as we like, and it may even help us reduce the number of conflicts that result due to shifts. For example, in Figure~\ref{figure_shift_subtrees_example}, if we would swap the green and purple nodes before shifting green, then the green shifting won't generate a red conflict. Therefore, we can only benefit from multiple conflicts in the same level, thanks to the degrees of freedom of ordering them in the most beneficial way, which would be according to their relative order in $\sigma$.}

By moving subtrees on the colored trie, fixing the levels top-down, we eventually reach a final coloring where the colors are in consecutive segments, according to the parts of $P$ ordered by $\sigma$. Since each original conflict might have inflated to at most $k-1$ conflicts, we get that $n \le \lambda(P) \cdot (k-1)$.

A more careful analysis reveals that we can save some, for instance, the root's conflict is ordered so it doesn't inflate at all. Similarly, moving around the second conflict also doesn't add new conflicts. The upper bound of added conflicts increases only when more segments exist, so the worst case is when the segments are added in different levels, such that fixing the order of the $i$-th conflict for $1 \le i \le k$ adds at most $\max(0,i-2)$ conflicts in the new coloring. Only for $i \ge k$ the shifting cost stabilizes on up-to $k-2$ additional conflicts. So we get:
$
n = \lambda(P) + ADDED
\le \lambda(P) + \sum_{i=1}^{k}{\max(0,i-2)} + \sum_{i=k+1}^{\lambda(P)}{(k-2)}
= \lambda(P) + \frac{(k-2)(k-1)}{2} + (\lambda(P)-k)\cdot (k-2)
= \lambda(P) \cdot (k-1) - \frac{(k-2)(k+1)}{2}
$.
\end{proof}

\begin{remark}
\label{remark_organization_of_bitmatcher_coloring_fixed_perm}
Observe that for $k=1$, we get from Lemma~\ref{lemma_any_order_k_approximation} that $n \le 1$, which is indeed the case when the tree is monochromatic. Moreover, for $k=2$ we get that $n \le \lambda(P)$. Since $\lambda(P) \le n$ as well, we conclude that $n = \lambda(P)$, which provides an alternative proof to Theorem~\ref{theorem_consecutive_k2}.
\end{remark}

If we consider a random permutation, we can improve Lemma~\ref{lemma_any_order_k_approximation} by a constant factor.

\begin{lemma}
\label{lemma_randomized_order_k_approximation}
Let $P$ be a partition of $2^W$ into $k$ parts and let $\sigma^* \in S_k$ denote a permutation that minimizes $N(P,\sigma)$. Then $\mathbb{E}[N(P,\sigma)] \le N(P,\sigma^*) \cdot \frac{k+1}{3} - \frac{(k-2)(k+1)}{6}$ where the expectation is over $\sigma$ chosen uniformly from $S_k$.
\end{lemma}

\begin{proof}
The proof relies on the same arguments as in Lemma~\ref{lemma_any_order_k_approximation}. The difference is that since the order of the parts is chosen at random, the expectation of the number of conflicts that are added for each re-ordering of subtrees is bounded as follows. For each conflicting node $v$ define by $X_v$ a random variable that tells how many conflicts are added as a result of shifting $v$ to its correct location according to a given permutation $\sigma$. If $v$ was generated as a child of the $j$-th segment, and should be shifted to the $i$-th segment ($1 \le i,j \le k$, $i \ne j$), then at most $|i-j| - 1$ new conflict will be generated. We subtract $1$ because moving $v$ to the edge of its parent's color segment doesn't add a conflict. Therefore, if there are $k$ segments, for $\sigma$ chosen uniformly over $S_k$ the locations of $i$ and $j$ are uniform over the $k(k-1)$ choices of pairs $i \ne j$, so:
$$\mathbb{E}_\sigma[X_v] \le \mathbb{E}_\sigma[|i-j| - 1] = \mathbb{E}_\sigma[|i-j|] - 1
$$
$$= \frac{1}{k(k-1)} \sum_{i=1}^{k}{ \sum_{(i\ne)j=1}^{k}{|i-j|} } - 1$$
Focusing on the double summation:
$$
\sum_{i=1}^{k}{ \sum_{(i\ne)j=1}^{k}{|i-j|} }
= \sum_{i=1}^{k}{ \Big ( \sum_{j=1}^{i}{(i-j)} + \sum_{j=i}^{k}{(j-i)} \Big )}
$$
$$
= \sum_{i=1}^{k}{ \Big ( \frac{i(i-1)}{2} + \frac{(k-i)(k-i+1)}{2} \Big ) }
$$
$$
= \frac{k^2 + k}{2} \sum_{i=1}^{k}{1} + \sum_{i=1}^{k}{i^2} - (1+k) \sum_{i=1}^{k}{i}
= \frac{(k-1)k(k+1)}{3}
$$
Therefore:
$$\mathbb{E}_\sigma[X_v] \le \frac{1}{k(k-1)} \cdot \frac{(k-1)k(k+1)}{3} - 1 = \frac{k-2}{3}$$

We note that for the first conflicts, the effective value of $k$ is smaller, as argued in Lemma~\ref{lemma_any_order_k_approximation}. Moreover, the first conflict doesn't have negative expectation, but simply $0$ added conflicts. Finally, while the event of where one conflict is added, or moved to, may affect future conflict events in that subtree, the expectation is linear and we get:
$$
\mathbb{E}_\sigma[n] = \lambda(P) + \mathbb{E}_\sigma[ADDED]
= \lambda(P) + \sum_{v \in conflicts}{\mathbb{E}_\sigma[X_v]}
$$
$$
\le \lambda(P) + \sum_{i=2}^{k}{\frac{i-2}{3}} + \sum_{i=k+1}^{\lambda(P)}{\frac{k-2}{3}}
$$
$$
=   \lambda(P) + \frac{1}{3} \cdot \frac{(k-2)(k-1)}{2} + (\lambda(P)-k) \cdot \frac{k-2}{3}
$$
$$
= \lambda(P) \cdot \frac{k+1}{3} - \frac{(k-2)(k+1)}{6}
$$
\end{proof}

\begin{remark}
\label{remark_organization_of_bitmatcher_coloring_uniform_perm}
Similar to Remark~\ref{remark_organization_of_bitmatcher_coloring_fixed_perm}, we can substitute $k=2$ to get an alternative proof of Theorem~\ref{theorem_consecutive_k2}. In this case, we get that the expectation is at most $\lambda(P)$, but since each individual permutation satisfies $N(P,\sigma) \ge \lambda(P)$, equality follows. For $k=1$, $\lambda(P) = 1$ and there is only one permutation, and indeed $1 \le n \le 1 \cdot \frac{2}{3} - \frac{(-1) \cdot 2}{6} = 1$ as expected.
\end{remark}

\begin{corollary}
\label{corollary_find_decent_permutation}
Let $P$ be a partition of $2^W$ into $k$ parts and let $\sigma^* \in S_k$ denote a permutation that minimizes $N(P,\sigma)$. Then we can find efficiently a permutation $\sigma' \in S_k$ such that: $N(P,\sigma') \le N(P,\sigma^*) \cdot \frac{k+1}{3} - \frac{(k-2)(k+1)}{6}$.
\end{corollary}

\begin{proof}
By Lemma~\ref{lemma_randomized_order_k_approximation} there must be a permutation that satisfies the claim. We find it through the following derandomization process: For a given partition $P$, let $\chi^*$ be the initial coloring with $\lambda(P)$ conflicts, not necessarily in segments. We have that $\mathbb{E}_\sigma[N(P,\sigma)] = \mathbb{E}_{i \in [k]} \big [ \mathbb{E}_{\sigma: \sigma_1 = i}[N(P,\sigma)] \big ] \le \frac{k+1}{3} - \frac{(k-2)(k+1)}{6}$. By computing each of the empiric expectation bounds for fixing $\sigma_1 = i$ (given $\chi^*$), we can fix $\sigma_1 = i_1$ for $i_1$ that minimizes the bound over $\mathbb{E}_{\sigma: \sigma_1 = i_1}[N(P,\sigma)]$.

The way to compute the conditional-expectation is not too different than the (unconditional) expectation. Revisit the proof of Lemma~\ref{lemma_randomized_order_k_approximation}, where we claimed that:
$\mathbb{E}_\sigma[X_v] \le \frac{1}{k(k-1)} \sum_{i=1}^{k}{ \sum_{(i\ne)j=1}^{k}{|i-j|} } - 1$. This expectation is easily computed for every node $v$ in conflict (in $O(k^2)$ time). Now, to condition the expectation on $\sigma_1 = c$ for a specific color $c$, we have two cases: (a) If $v$ or its parent are not colored by $c$, then $\mathbb{E}_\sigma[X_v] \le \frac{1}{(k-1)(k-2)} \sum_{i=2}^{k}{ \sum_{(i\ne)j=2}^{k}{|i-j|} } - 1$ (i.e. we rule-out options with $i=1$ or $j=1$); (b) If $v$ or its parent are colored by $c$, then $\mathbb{E}_\sigma[X_v] \le \frac{1}{k-1} \sum_{i=2}^{k}{|i-1|} - 1$ (i.e. the expectation only depends on the other color). We emphasize that all of the (initial) $\lambda(P)$ conflicts are known, so we can compute the total conditional-expectation by summing over the conditional-expectation of every such $\mathbb{E}_\sigma[X_v]$ depending on its case.

Next, we continue iteratively to fix the values of $\sigma_j$ for $j=2,3,\ldots,k$ conditioned on the values $\sigma_\ell = i_\ell$ for $1 \le \ell < j$. Computing the conditional-expectation of each $X_v$ remains the same, with a third possible case: when the segments of both $v$ and its parent are in known locations, $i^*$ and $j^*$ respectively (according to what we already fixed), then $X_v = |i^* - j^*| - 1$. Since each time we pick the smallest bound over the expectation, the resulting permutation $\sigma' = (i_1,i_2\ldots,i_k)$ satisfies $N(P,\sigma') \le N(P,\sigma^*) \cdot \frac{k+1}{3} - \frac{(k-2)(k+1)}{6}$.
\end{proof}

\begin{theorem}
\label{theorem_approximation_segments_order}
Let $P$ be a partition of $2^W$ into $k$ parts and let $\sigma^* \in S_k$ denote a permutation that minimizes $N(P,\sigma)$. Then for any $\sigma \in S_k$: $N(P,\sigma) \le N(P,\sigma^*) \cdot \min(k-1,W-\floor{\lg k} + 1)$. Also, $\mathbb{E}[N(P,\sigma)] \le N(P,\sigma^*) \cdot \min(\frac{k+1}{3},W-\floor{\lg k} + 1)$ where the expectation is over $\sigma$ chosen uniformly from $S_k$.
\end{theorem}
\begin{proof}
The first part of the claim follows as the minimum over Lemma~\ref{lemma_any_order_w_approximation} and Lemma~\ref{lemma_any_order_k_approximation}, and the second part follows as the minimum over Lemma~\ref{lemma_any_order_w_approximation} and Lemma~\ref{lemma_randomized_order_k_approximation}.
\end{proof}

Now we provide a simple lower bound of the approximation ratio. It is asymptotically equal to the upper bound of Lemma~\ref{lemma_any_order_k_approximation}, but differs by a constant factor.

\begin{theorem}
\label{theorem_approximation_segments_order_lower_bound_k_or_w}
Let $k$ and $W$ be such that $k-2 \le 2^{W-1}$. There exists a partition $P$ of $2^W$ into $k$ parts, and a permutation $\sigma \in S_k$ such that: $N(P,\sigma) \ge N(P,\sigma^*) \cdot \frac{\floor{\frac{W - \floor{\lg k}}{2}} \cdot (k-1) + 1}{\floor{\frac{W}{2}} + (k-1)}$ where $\sigma^*$ minimizes $N(P,\sigma')$ among $\sigma' \in S_k$.
\end{theorem}

\begin{proof}
We only need to show a construction of $P$ and $\sigma$. Also note that the requirement $k-2 \le 2^{W-1}$ is not too harsh: in general, it must be that $1 \le k \le 2^W$, and when $k > 2^{W-1}$ then obviously any order of the parts requires at least $k$ conflicts, and no more than $2^W < 2k$ so this case is not interesting.

We construct $P$ as follow: $p_1 + p_2 = 2^{W-1}$, and set $p_i$ for $i \ge 3$ to be either $2^{(W-1) - \floor{\lg k}}$ or $2^{(W-2) - \floor{\lg k}}$ such that the sum over $P$ is $2^W$. For example, if $k=4$ then $p_3 = p_4 = 2^{W-2}$, and if $k=5$ then $p_3 = 2^{W-2}$ and $p_4 = p_5 = 2^{W-3}$. If $W$ is odd we set $p_1 = \floor{\frac{2^{w-1}}{3}}$ and $p_2 = 2p_1 + 1$, and if $W$ is even we set $p_1 = \ceil{\frac{2^{w-1}}{3}}$ and $p_2 = 2p_1 - 1$. One can verify that $p_1 + p_2 = 2^{W-1}$, and that their binary representation is $01...0101$ and $10...1011$ or $01...011$ and $10...101$.

Consider $\sigma^* = \{ 1,2,\ldots,k\}$: since each of $p_i$ for $i \ge 3$ is a completely aligned subtree, then each of these segments causes a single conflict. As for $p_1,p_2$, both are a partition of the left subtree of size $2^{W-1}$, so the number of conflicts due to $p_1,p_2$ would be $\lambda([p_1,p_2]) = \floor{\frac{W}{2}} + 1$. Therefore overall we get $N(P,\sigma^*) = \floor{\frac{W}{2}} + (k-1)$.

Next, consider $\sigma = \{ 1, 3,4, \ldots, k, 2\}$. By running Algorithm~\ref{alg_consecutive_ranges} on the first two (bottom) levels when $P$ is ordered by $\sigma$, we find that $k-2$ conflicts are produced, and the lower bits are effectively truncated. For example, if the input is $[85, 64, 64, 128, 171]$ ($k=5$,$W=9$) then after two steps going up on the trie the segments are $[21, 16, 16, 32, 43]$ (as if still $k=5$ but $W=7$). This pattern continues for as long as none of the parts becomes $0$, which is at least $(W-2) - \floor{\lg k}$ levels, that is at least $\floor{\frac{(W-2) - \floor{\lg k}}{2}}$ pairs of steps. Then we count (at least) an extra conflict for each of the segments ($k$ in total) afterwards, so we get:
$N(P,\sigma) \ge \floor{\frac{(W-2) - \floor{\lg k}}{2}} \cdot (k-1) + k = \floor{\frac{W - \floor{\lg k}}{2}} \cdot (k-1) + 1$. Therefore, we found that:
$$\frac{N(P,\sigma)}{N(P,\sigma^*)} \ge \frac{\floor{\frac{W - \floor{\lg k}}{2}} \cdot (k-1) + 1}{\floor{\frac{W}{2}} + (k-1)}$$
\end{proof}

\begin{corollary}
\label{corollary_tight_approximation_bounds}
The upper and lower bounds of Theorem~\ref{theorem_approximation_segments_order} and Theorem~\ref{theorem_approximation_segments_order_lower_bound_k_or_w} are asymptotically tight (up to constants).
\end{corollary}

\begin{proof}
Let $P$ be a partition with $k$ parts whose sum is $2^W$. Denote the ratio $\rho \equiv \frac{N(P,\sigma^*)}{N(P,\sigma)}$ for short. If $k \ge 2^{W-1}$, since $N(P,\sigma^*) \ge k$ and $N(P,\sigma) \le 2^W$ we get that $\rho = O(1)$. Otherwise, by Theorem~\ref{theorem_approximation_segments_order}, $\rho = O(\min(k,W - \lg k))$ and by Theorem~\ref{theorem_approximation_segments_order_lower_bound_k_or_w}, $\rho \ge \frac{\floor{\frac{W-\floor{\lg k}}{2}} \cdot (k-1) + 1}{\floor{\frac{W}{2}} + (k-1)}$:
\begin{enumerate}[label=(\arabic*),noitemsep]
    \item If $k = O(W)$: $\rho = O(k)$ and $\rho = \Omega(k)$.
    \item If $W = O(k)$: $\rho = O(W - \lg k)$ and $\rho = \Omega(W - \lg k)$. \qedhere
\end{enumerate}
\end{proof}

Finally, we leave the question of finding (efficiently) an optimal order of the segments, open. That being said, we detail one greedy natural approach that fails, and a few more notes.

The most natural approach is probably the following greedy algorithm: given $k$ parts, we position them from left to right in a greedy order as follows: first we find the segment of length $a$ such that if the partition was $[a,2^W-a]$ we would end up with a minimum number of conflicts. After $a$ is fixed, we find the next segment of length $b$ such that $[a,b,2^W-(a+b)]$ is minimized. This process continues until all segments are positioned, and since we have $O(k)$ runs of Algorithm~\ref{alg_consecutive_ranges} in each step, and total of $O(k)$ steps, this is efficient.

In some sense, this greedy approach is an extreme version of the derandomization process of Corollary~\ref{corollary_find_decent_permutation}, since it doesn't try to re-arrange a coloring, but just figures-out everything from scratch. Unfortunately, the greedy algorithm can fail because there may be partial choices which seem equivalently good, but eventually one is much worst. On the other hand, if instead of sticking to a single option we maintain a list of all the optimal options up to this point, then the algorithm is no longer efficient, because the number of candidates might be exponential. Examples \ref{example_picking_a_loser}-\ref{example_exponential_blowup} below demonstrate these scenarios.

\begin{example}[Picking a loser]
\label{example_picking_a_loser}
Consider the partition $P=[1,\ldots,1,7,\ldots,7]$ for an equal number of $1$s and $7$s. Clearly, if we order the parts alternating between $1$ and $7$, the number of conflict is minimizes and is exactly one per part since $1,7$ compose together ``nicely'' to $8$ (a power of $2$). However, given the greedy choices, we may pick only $1$s at first, and end up with the order $[1,\ldots,1,7,\ldots,7]$, which will result in a sub-optimal number of conflicts.
\end{example}

\begin{example}[Exponential blow-up]
\label{example_exponential_blowup}
Consider the partition $P=[1,2^r-1,2,2^r-2,\ldots,2^{r-1},2^{r-1}]$ for some $r$ such that the sum of $P$ is a power of $2$. The greedy choice in the first step may pick any part of the form $2^i$ (since it adds a single conflict, which is best). In the next step, assume that we pick its counterpart $2^r - 2^i$ (again, it only adds a single conflict so this is greedily optimal). If we try to track all the possible candidates, after $r$ steps have been applied, the list of greedy candidates contains at least $\binom{r}{r/2}$ options, for picking any $\frac{r}{2}$ pairs. This expression is approximately $\frac{2^r}{\sqrt{\pi r}}$ (by Stirling's approximation), and since $r = \frac{k}{2}$ then this number is exponential in $k$.
\end{example}

We remark that while these two examples show the problem with the greedy algorithm, they don't rule out ``practically good'' results for ``practical partitions''. Another natural approach is to attempt to find a good permutation by starting from a Bit Matcher coloring and minimizing the number of added conflicts due to shifting subtrees. This is the exact approach that was used in Lemma~\ref{lemma_any_order_k_approximation} and Corollary~\ref{corollary_find_decent_permutation}, but the latter only guarantees an approximation factor of $O(k)$ which is not very exciting. Perhaps this idea can be used differently, or in combination with others, to determine a permutation that dramatically improves the upper bound.

Other ideas may include: Constructing $\sigma$ ``lazily'' to reduce subtree rearrangements. This method assures that the first conflict of each color doesn't incur additional conflicts due to shifting, but this doesn't help asymptotically.\footnote{It reduces the upper bound from $N(P,\sigma^*) \cdot (k-1) - \frac{(k-2)(k+1)}{2}$, to $N(P,\sigma) \le N(P,\sigma^*) \cdot (k-1) - k(k-2)$ which is negligible.} One may consider to choose $\sigma$ such that the more conflicts a color has the closer it is to the center of $\sigma$, such that the number of shifts required to move these conflicts into place is reduced on average. A more sophisticated approach would be to consider the statistics over pairs of colors in the conflicts (parent-child colors), in order to put colors that conflict most often closer together.

\section{Conclusions}
\label{section_conclusions}

In this paper we studied the trade-off between having a smallest LPM TCAM representation for a given partition, versus the fragmentation of this representation to consecutive segments of addresses. More concretely, we showed that in general not all partitions of $k$ parts can be realized with the minimum number of LPM TCAM rules such that every part is allocated as a single consecutive range of addresses. We then showed that this generalizes such that for any maximum fragmentation $m$, one must either fragment the domain of some of the targets to more than $m$ segments, or use more than the minimum number of rules that are required to realize the partition.

Then, we proceeded to study the case where a partition must be realized in a single segment per target, and presented a scheme to construct a set of TCAM rules that is at most $\min(\frac{k+1}{3},W-\floor{\lg k} + 1)$ times larger than the minimal representation (given a best ordering of the segment).

It would be interesting to find and prove better approximations, or have some guarantees on the resulting TCAM size when each target is mapped to a fixed number of segments (the simplest case-study is when there is a single segment per target). Conversely, if finding the best ordering of segments is hard, it would be interesting to find a hardness reduction, or even a direct proof, to show that.

\bibliographystyle{IEEEtran}
\bibliography{reference}

\end{document}